\documentclass[10pt,twocolumn]{article} 
\usepackage{amsmath,epsf,amssymb,cite,pifont,amsthm, mathrsfs,epsfig,  bbm, amsthm,  setspace}
\input xy
\xyoption{all}
\graphicspath{ {Figures/} }
\newtheorem{thm}{Theorem}

\newtheorem{lemma}{Lemma}

\newtheorem{defn}{Definition}

\newcommand{\DD}{\mathcal D}

\newcommand{\FF}{\mathcal F}

\renewcommand{\P}{\mathbb P}

\newcommand{\R}{\mathbb R}
\newcommand{\RR}{\mathcal R}

\newcommand{\Z}{\mathbb Z}

\newcommand{\iso}{\cong}

\newcommand{\inv}{^{-1}}

\renewcommand{\phi}{\varphi}

\newcommand{\Cech}{\v{C}ech }

\newcommand{\HP}{\HP \textrm{P}}

\newcommand{\low}{\textrm{low}}
\newcommand{\Lk}{\textrm{Lk}}
\newcommand{\St}{\textrm{St}}

\newcommand{\RF}{\mathcal{R},\mathcal{F}}
\newcommand{\Ch}{\textrm{Ch}}

\newcommand{\norm}[1]{\| #1 \|}

\newcommand{\mindeath}{d}

\newcommand{\death}{r}
\newcommand{\sizeset}{k}
\newcommand{\chiint}{n}

\title{Failure Filtrations for Fenced Sensor Networks
\footnote{Research is partially supported by NSF under grants 
NSF-DMS-09-43760 and
NSF-DMS-10-45153, 
AFOSR  under grant  FA9550-10-1-0436, 
and
NIH under grant K25-AI079404.
}}
\author{Elizabeth Munch\footnote{Department of Mathematics, Duke University, Durham, NC},  Michael Shapiro\footnote{Department of Pathology, Tufts University, Boston, MA}, John Harer\footnote{Departments of Mathematics, Computer Science, and Electrical and Computer Engineering, Program in Computational Biology and Bioinformatics, Duke University, Durham, NC}}
\date{\today}

\begin{document}
\maketitle

\begin{abstract}

In this paper we consider the question of sensor network coverage for  a 2-dimensional domain.
We seek to compute the probability  that a set of sensors fails to cover given only
non-metric, local (who is talking to whom) information
and a probability distribution of failure of each node.
This builds on the work of de Silva and Ghrist who analyzed this
problem in the deterministic situation.
We first show that a it is part of a slightly larger class of problems which is \#P-complete, and thus fast algorithms likely do not exist unless P$=$NP.
We then give a deterministic algorithm which is feasible in the case of a small set of sensors,
and give a dynamic algorithm for an arbitrary set of sensors failing over time which utilizes a new criterion for coverage based on the one proposed by de Silva and Ghrist.
These algorithms build on the theory of topological persistence
\cite{Edelsbrunner2010}.

\end{abstract}

\section{Introduction}


The newly emerged field of Computational Topology \cite{Edelsbrunner2010} continues to find ever
increasing areas of application.
Perhaps its most significant application so far has been in the use of topological data analysis (TDA) on a wide variety
of datasets   \cite{Edelsbrunner2010} \cite{Carlsson2009} \cite{Chazal2009},
and has also been used effectively to find
structure in images \cite{Carlsson2008}\cite{Edelsbrunner2009},
shape in proteins and protein complexes \cite{Agarwal2006}\cite{Ban2004}\cite{Headd2007}
and in many other areas.
Recently, it was applied to sensor networks in \cite{Ghrist2005}, \cite{DeSilva2006}, \cite{DeSilva2007},\cite{Tahbaz-Salehi2010}
and the current paper is an extension of \cite{DeSilva2006}. 

Topology enters the study of sensor network when we consider questions like coverage.
When does a set of sensors effectively monitor a region and when are there gaps?
Phrasing this geometrically,
we start with a set of sensors $\chi$ in a domain $\Delta \subset \mathbb{R}^2$
where each can detect objects in a circular region of fixed radius $r_c$,
and we ask if the union of these discs covers all of $\Delta$.
This problem has been studied quite a bit,
but previous to \cite{DeSilva2006}, most work fell into one of two groups -
approaches that utilized geometric analysis to obtain an exact answer
and those that sought a non-deterministic approximation but assumed
significant capabilities of the sensors.
For a survey of the literature, see \cite{Yick2008}.

The former approach
requires a great deal of prior knowledge about the geometry of the domain and the exact location of the sensors.
The latter,
does not require this exactness, but often requires a uniform distribution of nodes or a high level
of intelligence in the
sensors.
The main contribution of  \cite{DeSilva2006} was a criterion for
coverage that  requires none of these things.

In the current paper, we take a middle ground and address  the question of computing  the probability of failure of
the criterion of \cite{DeSilva2006} given the probability of failure of each sensor.
We  show that a computing the probability of failure for a generalized set of complexes is   NP-hard, 
but we give an algorithm which can be used to solve small instances of the problem,
and an alternative, dynamic algorithm to give an early warning of potential failure.

\subsubsection*{Outline. }
The paper is organized as follows.
Section \ref{S: Persistent Homology} gives an introduction to Rips complexes and persistent homology,
both of which will be used extensively in this paper.
In Section \ref{S: Coverage Criterion}, we  summarize the problem and results of \cite{DeSilva2006}.
Section \ref{S: Sensor Failure} adds the assumption that sensors have a probability of failure, and Section
\ref{S: Complexity} discusses the complexity issues of determining the probability that there is coverage of
the domain.
Section \ref{S: Deterministic Algorithm} presents a deterministic algorithm for those times when
the set of sensors is small enough, and Section \ref{S: Monitored System} gives a
dynamic algorithm for use when the set of sensors is too large.

\section{Rips Complexes and Persistent Homology}\label{S: Persistent Homology}

Let $\chi$ be a set of points in $\R^2$ and suppose that $r>0$ is given.
We are interested in the topology of $W_r(\chi)$, the union of balls of radius $r$ about the points
of $\chi$.
One can build a variety of complexes with vertices $\chi$ that capture this topology,
the simplest and most intuitive is the \Cech complex which has a simplex
$<v_0, \ldots, v_k>$ whenever the balls of radius $r/2$ about the $v_i$ have a non-trivial
intersection.
The nerve lemma tells us that \Cech does indeed capture the topology of $W_r(\chi)$,
but it can be difficult to compute.
In particular, it requires one to know the exact location of the points of $\chi$, a luxury
that we do not have in this case.
We must therefore use an approximation known as the Rips complex for the problem
at hand.

The Rips complex $\RR$  has a $d$-dimensional simplex $\sigma$
whenever $\norm{v_i-v_j}<r$ for every pair of vertices $v_i,v_j \in \sigma$.
Unfortunately, the Rips complex does not retain the homotopy type of $W_r(\chi)$,
but what is lost in topological data is made up for in ease of computation.
Furthermore, the only  information necessary
to build the Rips complex is the set of pairs of points whose distance is below a the prescribed threshold $r$.

Since we will be considering filtrations of our simplicial complex and looking at how the topology
changes with the change of the simplicial complex, a brief review of persistent homology is in order \cite{Edelsbrunner2010}.

We begin with a filtration of a simplicial complex $\RR$, given by a series of inclusions
\begin{center}
 \includegraphics[scale=.85]{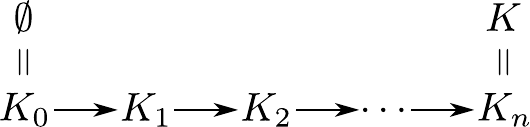}
\end{center}
which induces maps on homology

\begin{center}
 \includegraphics[scale=.85]{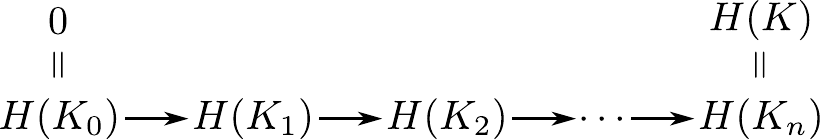}
\end{center}
We will use homology with $\Z_2$ coefficients for the entirety of the paper.
Consider how the sequence of homology groups changes as the simplicial complex changes.
A class $[\alpha] \in H_p(K_i)$ is {\em born} at $K_i$ if it is not in the image of the map
$H_p(K_{i-1}) \to H_p(K_i)$.
This class {\em dies} entering $K_j$ if once there it merges with an older class,
i.e. lies in the image of the map $H_p(K_{i-1}) \to H_p(K_j)$.
See Figure \ref{F:BirthDeathClass}.

\begin{figure}
 \begin{center}
\includegraphics[scale = .8]{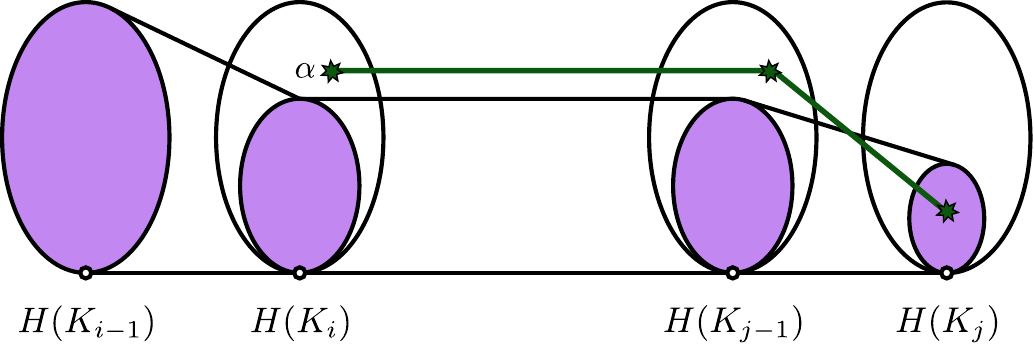}
 \end{center}
 \caption{A visual representation of births and deaths of homology classes.
 The class $\alpha$ is born at $K_i$ because it is not in the image of the map from $H(K_{i-1})$.
 It dies entering $K_j$ because it was still not in the image of $H(K_{i-1})$ at $H(K_{j-1})$,
 but has merged with an older class upon entering $H(K_j)$.}
\label{F:BirthDeathClass}
\end{figure}

To determine when classes are born and die in the filtration, we build the boundary matrix $D$.
This is a square matrix with a row and column for each simplex in $K$,
ordered with respect to the filtration (which  ensures that a simplex comes after all of its faces).
$D$ is a $0$-$1$ matrix which has a $1$ in location $D[i,j]$ if and only if the $i$th simplex is a
face of the $j$th.
Applying the persistence algorithm as in \cite{Edelsbrunner2010} yields a reduced matrix $R=DV$,
where $D$ is the boundary matrix and $V$ is an elementary matrix storing the column operations
performed on $D$ during the persistence algorithm.
Here, a reduced matrix is one in which every column is either completely zero, or the lowest $1$
in the column is not the same as the lowest 1 in any other column.
We write $\low_R(j)=i$ if the lowest 1 in column $j$ of matrix $R$ is in row $i$.

To read the births and deaths from the matrix, note that a $p$-dimensional class is born with addition of a
dimension $p$ simplex $\sigma$ if the column corresponding to $\sigma$ is completely zero.
A representative for the class that is born is stored in the corresponding column of $V$.
The class born at $\sigma$ dies with the addition of $\tau$ if the lowest one of the column corresponding to
$\tau$ is in the row corresponding to $\sigma$,
in which case $\sigma$ and $\tau$ are paired.
If the addition of a simplex $\sigma$ gives birth to a class, it is called a \textit{positive simplex}.
Similarly, if the addition of a simplex $\tau$ gives death to a class, it is called a \textit{negative simplex}.

The de Silva-Ghrist criterion requires us to work with persistent \textit{relative} homology.
In this case, we take a pair $(\RR,\FF)$, where $\FF$ is a subcomplex of $\RR$, and a filtration
\begin{center}
 \includegraphics[scale=.85]{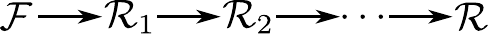}
\end{center}
 inducing maps on relative homology
\begin{center}
 \includegraphics[scale=.85]{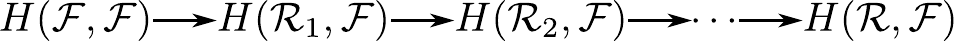}
\end{center}
and consider births and deaths in the usual way.
This requires a slight modification of the boundary matrix by reordering the rows so that those
simplices which are in the subspace $\FF$ are moved to the top of the matrix prior to performing the persistence algorithm.
Then the addition of $\sigma$ gives birth to a class if its column is either zero,
or the lowest one in its column corresponds to a simplex in $\FF$.
Simplex $\tau$ gives death to a class if its lowest one corresponds to a simplex which is not in $\FF$.
Other than these distinctions, computing persistent homology in the relative case is the same as in the absolute case.

\section{The Coverage Criterion}\label{S: Coverage Criterion}

Working in a simply-connected domain in the plane,
suppose that we have a set of sensors with a fixed radius of coverage.
Our goal is check that these sensors cover the whole domain.
What makes this a challenging problem is that
{\em we do not assume that we know the locations of the sensors}.
This means that standard geometric techniques are not applicable.
Instead we turn to topology to answer the coverage question by building
the {\em Rips complex} on the set of sensors,
thought of as points in the plane.
We can then use homology to check for holes in the coverage.

Let   $\chi$ be the set of points corresponding to the location of the set of sensors in a compact connected domain $\Delta\subset \R^2$
which has a piecewise linear boundary.
Suppose that each sensor has a fixed coverage radius $r_c > 0$.
The question is whether every point in $\Delta $ lies within distance $r_c$ of some sensor in $\chi$.
We do not use the distance $r_c$ to build the Rips complex, instead we add an additional capability to each sensor.
Let $r_b > 0$ be fixed, with $r_b \leq \sqrt{3} r_c$ for technical reasons.
Each sensor is given a unique identification number to broadcast.
If another node is within distance $r_b$, it can hear the signal and identify the ID
number, but it has no information about the location of the broadcaster.
In particular, it does not know its direction or its exact distance,  only that that distance
is less than $r_b$.
Whenever two sensors can hear each others' identification number, an edge is
placed in the Rips complex.
Higher dimensional simplices are then added when all of their faces are already there.

The boundary of the domain $\Delta$ is taken to be
piecewise linear with a sensor at each of its vertices.
The boundary is called the
 \textbf{fence}, and
each  node in the fence knows the identification number of its two
fence neighbors, both of which are  within distance $r_b$.

Summarizing, following \cite{DeSilva2006},
 the assumptions are:
\begin{enumerate}
	\item Nodes $\chi$ broadcast their unique ID numbers.
	          Each node can detect the identity of any node within broadcast radius $r_b$.
	\item Nodes have radially symmetric covering domains of cover radius $r_c \geq r_b/\sqrt{3}$.
	\item Nodes lie in a compact connected domain $\Delta \subset \R^2$ whose boundary
	         $\partial \Delta $ is connected and piecewise-linear with vertices marked fence
	         nodes $\chi_f$.
		  Non-fence nodes are called interior nodes,  and denoted $\chi_{int}$.
	\item Fence nodes $\chi_f$ are ordered cyclically and each $v \in \chi_f$ knows the
	        identities of its two neighbors on $\partial \Delta $.
	        These neighbors both lie within distance $r_b$ of $v$.
\end{enumerate}
Based on this information, build the Rips complex $\mathcal{R}$ with the fence $\mathcal{F}$ as a subcomplex.
With this setup, de Silva and Ghrist in \cite{DeSilva2006}
give their controlled boundary criterion for coverage:
\begin{thm}[de Silva/Ghrist Criterion (dS-G)] \label{Criterion}
   If there is a nontrivial element of the relative homology group
   $H_2(\mathcal{R}, \mathcal{F})$ which maps to a nonzero class under the
   connecting homomorphism $H_2(\mathcal{R}, \mathcal{F}) \to H_1(\mathcal{F})$,
   then the union of the disks of radius $r_c$ about the nodes contains all of $\Delta $.
\end{thm}

The class $[\alpha] \in H_2(\RF)$ is \textbf{fundamental} if it satisfies the criterion of
Theorem \ref{Criterion},
but we stress that when there is such an element it is not necessarily unique.
The term \textbf{absolute cycle} will be used for a class in $H_2(\RF)$ that comes
from $H_2(\RR)$, which is equivalent to saying that it maps to $0$
under the connecting homomorphism.

The assumption that $r_c \geq r_b/\sqrt{3}$ is required to compensate for the fact that the Rips complex
does not accurately reflect the topology of the cover.
While this bound promises that holes in the cover appear also as holes in the Rips complex,
we can still create examples where phantom holes appear in the Rips complex even though no hole exists in the cover itself.
In a perfect world, this theory would be built on \Cech complexes,
however the lack of location data for the nodes makes this method impossible.


\section{Sensor Failure} \label{S: Sensor Failure}
Over time, sensors have a likelihood of failure which increases the longer the system is in place,
caused perhaps by malicious actions, environmental conditions or mechanical failure.
As nodes fail, there are two possible effects on the coverage:
either the death of a subset of nodes
creates a hole in the Rips complex,
or the removal of the nodes does not affect the existence of a fundamental class.
Once again, we emphasize
 that we are specifically not looking for the probability of failure of the \textit{cover}
over time, just the failure of the dS-G criterion.
For this reason, we also assume that only interior nodes can fail.
The loss of a fence node causes instant failure of the dS-G criterion,
so there is nothing to check in this case.

Let us start with a Rips complex pair $(\RR,\FF)$ built from a set of nodes $\chi$.
At time $t=0$, we assume we have a fundamental class $[\alpha] \in H_2(\RF)$.
If a set of interior sensors $B \subset \chi_{int}$ fails, any simplex in $\RR$
that has a vertex in the set $B$ is lost.
Therefore this subcomplex, $\RR_B$, can be thought of as the largest subcomplex of $\RR$ that has
 $\chi-B$ as its vertices.
We could then determine whether $\RR_B$ fails the dS-G criterion by looking for a fundamental class in
$H_2(\RR_B,\FF)$, but this is a very narrow view of the problem.
Much more information is available in a filtration that we will now construct.
Note that it will  contain $\RR_B$ as one of its subcomplexes.

Let $|\chi_{int}|=\chiint$.
Order the nodes so that
$B=\left\{ v_1,\cdots,v_\sizeset \right\}$ and $\chi_{int}-B=\left\{ v_{\sizeset+1},\cdots,v_\chiint \right\}$.
Let $V_{i}=\left\{ v_1,\cdots,v_i \right\}$ so that $V_\sizeset=B$,
and let $\RR_i=\RR_{V_i}$ be the maximal subcomplex of $\RR$ with vertices $\chi - V_i$.
Then the filtration
\begin{center}
 \includegraphics[scale=.85]{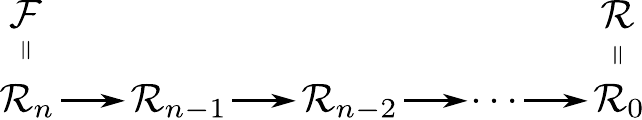}
\end{center}
induces maps on relative homology
\begin{center}
 \includegraphics[scale=.85]{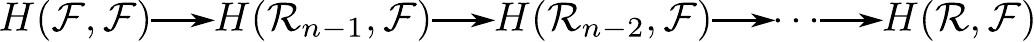}
\end{center}
An example of this filtration is illustrated in Figure \ref{fig: rips complex example}.
\begin{figure}[ht]
\begin{center}
\includegraphics{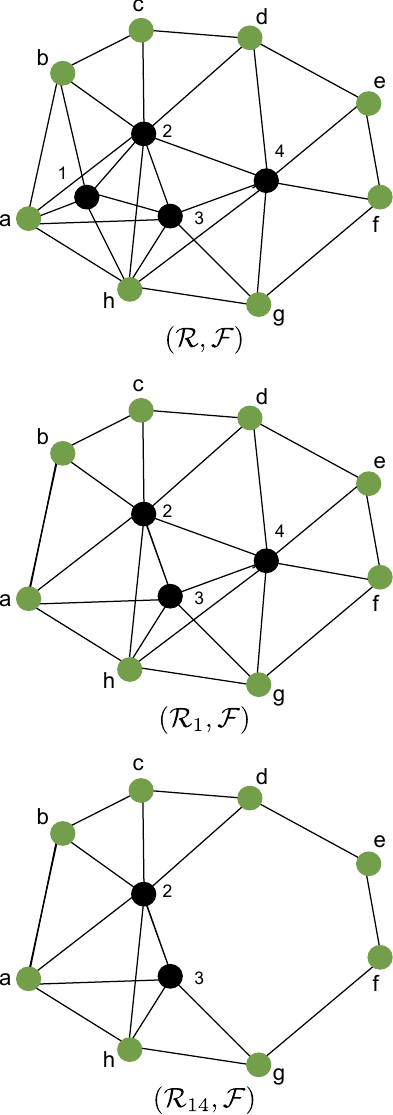}
\end{center}
\caption{ A small scale example of the Rips complex $\RR$ built from a set of points in the plane.
The outer ring of nodes labeled with letters is the fence $\FF$.} \label{fig: rips complex example}
\end{figure}
Intuitively, we expect that discovering a fundamental element at any point in this sequence implies
that there is a fundamental element in any subsequent group.

 \begin{lemma}\label{L:map}
Let $A \subset B$ be subsets of $\chi_{int}$.  Then if $[\beta] \in H_2(\RR_B,\FF)$ is fundamental,  its image under the map
$$
\xymatrix{
H_2(\RR_B,\FF)  \ar[r]^{i_*}& H_2(\RR_A,\FF)
}
$$
is also  fundamental.
\end{lemma}

  \begin{proof}
 Consider the commutative diagram
 \begin{equation}\label{D:lemma1}
	 \xymatrix{
	 H_2(\RR_B,\FF) \ar[r]^{i_*}\ar[d]_\partial & H_2(\RR_A,\FF) \ar[d]^\partial\\
	 H_1(\FF) \ar[r]^= & H_1(\FF)
	 }
 \end{equation}
 where the horizontal maps are induced by the inclusion $\RR_B \subset \RR_A$ and the vertical maps are the boundary maps.
Since $[\beta] \in H_2(\RR_B,\FF)$ is fundamental,  $\partial [\beta]\neq 0$ in $H_1(\FF)$.
Since the diagram commutes, $\partial i_*[\beta] = \partial [\beta]$ and hence is nonzero.
This also implies that $i_*[\beta]$ is nonzero, so it is a fundamental element of $H_2(\RR_A,\FF)$.
  \end{proof}

This lemma shows that if  $\RR_B$ passes the dS-G criterion,
then $\RR_A$ passes the dS-G criterion for all $A\subset B$.
But it also shows that if $\RR_B$ fails the dS-G criterion, then $\RR_A$ fails the dS-G  criterion for all
$A \supset B$.
Thus if  $B$ is a death set, and $B \subset A$,
Lemma \ref{L:map} implies that $A$ is also a death set.
This leads us to make the following definitions:

\begin{defn}
A set $B \subset \chi_{int}$ whose removal causes failure of the dS-G criterion is called a
\textbf{death set}.
A death set $B \subset \chi_{int}$ is a \textbf{minimal death set} if no subset of $B$ is itself a death set.
\end{defn}

\begin{defn}
If the removal of $B$ does not cause failure, we call $B$ a \textbf{cake set}.
A cake set $B \subset \chi_{int}$ is a \textbf{maximal cake set} if no superset of $B$ is also a cake set.
\end{defn}
As we will show in section \ref{S: Probability of Failure}, minimal death sets are directly related to the failure of the dS-G criterion.
However, we first look at the issues arising from the complexity of the problem.


\section{Complexity Issues}\label{S: Complexity}

The first issue to address is whether this problem is computationally complex.
In this section, we will show that  in fact it is part of a slightly larger group of problems which are  NP-hard, more specifically \#P-complete.

\subsection{Use of the 2-skeleton} \label{S: 2-skeleton}
A simplifying step is to work with the 2-skeleton $\RR^2$ of $\RR$ rather than the full complex.
To justify this, notice that passing to the 2-skeleton does not affect our observance of the dS-G criterion:

\begin{lemma}
The dS-G criterion is satisfied for $\RR$ if and only if it is satisfied for $\RR^2$,
where $\RR^2$ is the 2-skeleton.
\label{L: 2-skeleton}
\end{lemma}
\begin{proof}
Consider the following diagram built from the long exact sequences for the pairs $(\RR,\FF)$
and $(\RR^2,\FF)$ and the maps induced by the inclusion $\RR^2 \to \RR$:
\begin{equation*}
\xymatrix{
H_2(\RR^2,\FF) \ar[d]_{i_*}\ar[r] & H_1(\FF)\ar[d]_{=} \ar[r] & H_1(\RR^2)\ar[d]\\
H_2(\RF) \ar[r] & H_1(\FF) \ar[r] & H_1(\RR)
}
\end{equation*}
If the dS-G criterion is satisfied for $\RR^2$, then there is an $\alpha \in H_2(\RR^2,\FF)$
such that $\partial(\alpha)$ is nonzero.
Clearly $i_*(\alpha)$ also satisfies the dS-G criterion since $\partial(i_*(\alpha))=\partial \alpha$.

Now assume that the dS-G criterion is satisfied for $\RR$.
Then there is a $\beta \in H_2(\RR,\FF)$ such that $\partial\beta$ is nonzero in $H_1(\FF)$.
Here it is important to note that since $\RR^2$ is the 2-skeleton of $\RR$,
$H_1(\RR)=H_1(\RR^2)$.
As the top and bottom rows are exact with the last two groups equal,
and since $\partial \beta \in H_1(\FF)$ maps to zero in $H_1(\RR)$,
it follows that it also maps to zero in $H_1(\RR^2)$.
Because the top row is exact, there is an $\alpha \in H_2(\RR^2,\FF)$ which maps to $\partial \beta$,
and hence satisfies the dS-G criterion.
\end{proof}

This lemma implies that the sets of death sets, minimal death sets, cake sets,  and maximal cake sets are equivalent
to their counterparts when computed in the 2-skeleton.
It also implies that the probability of failure of the dS-G criterion in the full Rips complex
and the probability of failure of the dS-G criterion in the 2-skeleton are the same.
And lastly, in $\RR^2$ there is exactly one cycle representing each homology class ($Z_2 = H_2$).
For these reasons we will simplify notation and write $\RR$
for the 2-skeleton of the Rips complex for the remainder of the paper.


\subsection{\#P-Complete}\label{S: SharpP-Complete}

The class of problems defined as \#P-complete was introduced by Valiant in \cite{Valiant};
they form a specific class of NP-hard problems.
Typically, \#P-complete problems are concerned with counting \textit{how many} of something
exists  whereas general NP problems just ask \textit{if} something exists.
Problems which are \#P-complete likely do not have polynomial time algorithms.

To show that a problem is \#P-complete, we reduce one difficult problem to another.
Reducing problem $A$ to problem $B$ means that we take any instance of problem $A$,
use it to create an instance of problem $B$, and conclude that the answer to solving problem $B$
gives an answer to problem $A$.
To prove NP-completeness or \#P-completeness,
both turning an instance of problem $A$ into an instance of problem $B$
and returning the answer to problem $A$ given the solution to problem $B$ must
be done in polynomial time.
If we can reduce $A$ to $B$ in polynomial time,  we write $A \leq_P B$.

A reduction from $A$ to $B$ is called \textit{parsimonious} if the number of solutions for $A$
is in one-to-one correspondence to solutions for $B$.
This is an important property for proving that problems are \#P-complete since we need to be
able to count the number of solutions of $A$ based on the number of solutions of $B$.

In order to show that our sensor network problem is \#P-complete,
we need to find a polynomial time, parsimonious reduction from a \#P-complete problem to our problem.

It has been known for several decades that the computer science problem of {\bf network reliability}
is \#P-complete \cite{Garey1979,Colbourn1987}.
We will specifically work with the two terminal network reliability problem as defined in \cite{Garey1979}.
An instance of the problem is a graph $G=(V,E)$ with marked vertices $\left\{ s,t \right\}$,
a rational failure probability $p_e$, $0 \leq p_e \leq 1$ for each edge $e \in E$,
and a positive rational number $q\leq 1$.
Then, assuming edge failures are independent of one another,
we ask whether the probability that $s$ and $t$ have a path with no failed edge
is greater than or equal to $q$.
The fact that this problem is hard in the class of counting problems comes from needing to count the possible paths from $s$ to $t$ when determining the probability of failure.

This problem has striking similarities to ours, and the closeness is even more pronounced
when we look at it in the following way.
Considering the graph $G$ as a one-dimensional simplicial complex, a path in $G$ with
endpoints at $s$ and $t$ is a fundamental class in $H_1(G,\left\{ s,t \right\})$,
where fundamental means that the boundary of the class is homologous to $[s]+[t]$ in
$H_0(\left\{ s,t \right\})$.

Our goal is to reduce network reliability to our problem, which we will therefore call
{\em 2-dimensional network reliability}.
An instance of the problem is a simplicial complex $X$ with a subcomplex
$Y$ that is homeomorphic to $S^1$.
We also have rational probabilities of failure $p_v$, $0 \leq p_v \leq 1$, on the  vertices not in $Y$,
and a value $0<q\leq 1$.
We ask the following question:
{\em Given the fact that failures of vertices are independent of each other, is the probability
that we have a fundamental class $\alpha \in H_2(X,Y)$ at least $q$?}

Notice that our definition of the problem takes no account of the geometry inherent in the
originally defined problem
as we are ignoring the fact that we obtained this simplicial complex from a set of points in $\R^2$, and thus we are proving a larger class of problems to be \#P-complete.
To prove that 2-dimensional network reliability is \#P-complete, we must take an instance of the
1-dimensional network reliability problem, turn it into an instance of the 2-dimensional case in
polynomial time, take the solution given there and turn it into an answer to the 1-dimensional
case in polynomial time.

\begin{thm}
2-dimensional network reliability is \#P-complete.
\end{thm}

\begin{proof}
Consider a finite graph $G$ with vertex set $V$ and  probability of failure
$p_e$ given on each edge $e$.
We will construct a 2-dimensional simplicial complex $X$ with a subcomplex $Y \iso S^1$
 so that there is a one-to-one correspondence between paths from $s$ to $t$ in $G$
  and fundamental classes of $H_2(X,Y)$.
This correspondence will also preserve the probability of failure of the class,
so this will imply that the probability of failure in the 1-dimensional case can be computed by
determining the probability of failure in the 2-dimensional case.

Suppose that $G$ has $n$ vertices.
Order these vertices by choosing a map
$r:V \to \R$ that sends each vertex to a distinct integer in
$\left\{ 1,\cdots,n \right\}$, with $r(s)=1$ and $r(t)=n$.
Extend $r$ to all of $G$ by linear interpolation over each edge,
and subdivide $G$ by adding vertices at all points of $r^{-1}( \{ 1,\cdots,n \} )$.
Call the result $G'$, the map $r': G' \to $ is now piecewise linear.
If an edge $e$ of $G$ is subdivided into $k$ subedges in $G'$,
we set the probability of failure for one of the subedges equal to $p_e$ and the rest equal to 0.  
See Figure [\ref{fig:graphexample}] for an example of building this graph.

Form the complex
\begin{equation*}
G'\times I / \sim
\end{equation*}
where $(x,z) \sim (x',z')$ iff $r(x)=r(x')$ and either $z=z'=1$ or $z=z'=0$.
Note that this collapses the top and  bottom graphs each onto a separate copy of the interval
$[r(s), r(t)]$.
To make this a true simplicial complex, divide each rectangle of the form $e \times I$
into triangles by placing a vertex at the barycenter and adding the obvious four new edges
and four new triangles.
The resulting complex will be called $X$.
Then define $Y$ to be the subcomplex
$G'\times \left\{ 0,1 \right\}/\sim$ together with the two edges $s \times I$ and $t \times I$;
$Y$ is homeomorphic to $S^1$ by construction.
(In Figure [\ref{fig:graphexample}] we
have not subdivided the rectangles to keep the picture uncluttered.)

\begin{figure}[h]
\begin{center}
\includegraphics[scale=.5]{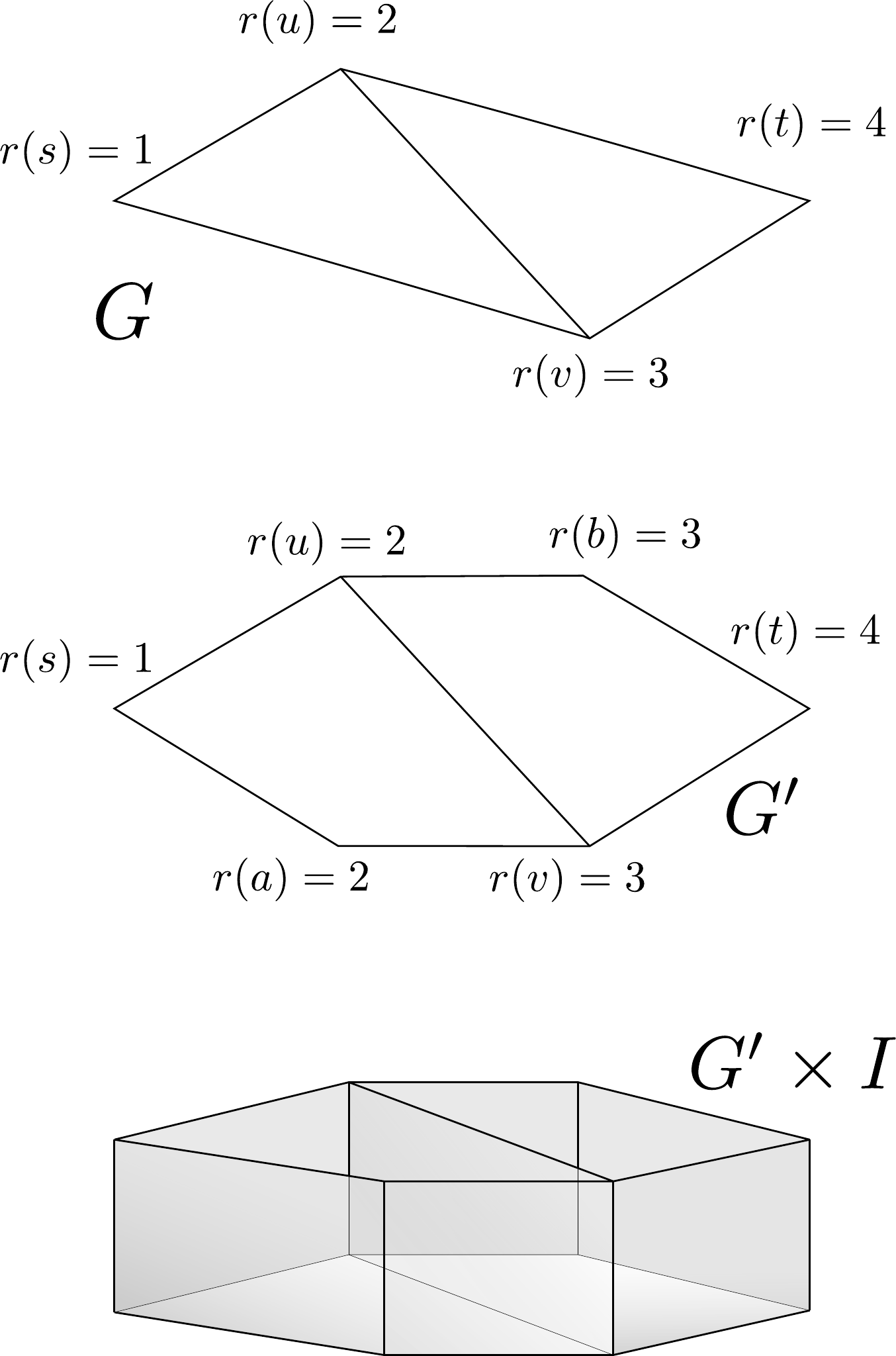}
\end{center}
\caption{$G$ is the instance of the 1-dimensional Network Reliability problem with an ordering $r$ placed on the vertices.
This map can be extended by linear interpolation over each edge.
We then subdivide the edges at all points of $r \inv (\{1,\cdots,n\})$.
Finally, we define the complex $X = G'\times I / \sim$ where $(x,z) \sim (x',z')$ iff $r(x)=r(x')$
and either $z=z'=1$ or $z=z'=0$.  }
\label{fig:graphexample}
\end{figure}

Set the probability of
failure of vertices that were added to the centers of the rectangles
equal to the probability of failure of the edge of $G'$ from which they arose.
Notice that failure of one of these vertices
leads to removal of the interior of the corresponding rectangle.

It is obvious that each path in $G$ gives rise to a fundamental class in $H_2(X,Y)$.
To prove the opposite, first recall that since $X$ has no $3$-simplices, each fundamental
class in $H_2(X,Y)$ has a unique representative cycle ($B_2(X,Y) = 0$).
Furthermore, since we are using homology with $\Z_2$ coefficients,
each class is simply a subset of the set of $2$-simplices in $X$.
A relative cycle $\alpha \in Z_2(X,Y)$ has the added property that an even number of
$2$ simplices in $\alpha$ contain any edge of $X-Y$
and a fundamental class must have $\partial \alpha$ equal to the sum of all of
the simplices of $Y$.
Notice also that if $\alpha$ contains any  2-simplex from a rectangle, it must contain all
four 2-simplices from that rectangle,
so it is equivalent to think of $\alpha$ as a set of rectangles from before the subdivision.

Our conclusion now follows easily.
Every rectangle in $\alpha$ determines a unique edge of $G'$.
Since $\FF$ is covered exactly once, no two edges of $G'$ that
are equivalent under $\sim$ can occur as edges of rectangles in $\alpha$.
Since $\alpha$ is a relative cycle, each vertical edge must lie on either
$0$ or $2$ rectangles, so the edges patch together to give a path from
$s$ to $t$.
Hence there is a one to one correspondence between the two sets.

Since we set up each rectangle to have an equal probability to that of its corresponding edge,
the probability that a fundamental class is still in $H_2(X,Y)$ is equal to the probability that
the corresponding path in $G$ is still functioning.
Thus, if we could compute the probability of failure in $X$ in a reasonable time frame, the solution would
give the probability of failure in $G$ in a reasonable time frame.
Since the latter problem is \#P-complete, our 2-dimensional version is also \#P-complete.

\end{proof}


\section{A Deterministic Algorithm}\label{S: Deterministic Algorithm}

Now that we know that the general problem is \#P-complete,
we strive to find ways to work around the computational complexity issues.
We will first show that, given a set of sensors which is relatively small
or at the very least relatively sparse in the domain, we can write a deterministic
algorithm to compute the probability of failure of the system.
In section \ref{S: Monitored System}, we will  consider the modified
problem of predicting failure as sensors in the system fail.


\subsection{The Hasse Diagram}\label{S: Hasse}

Consider a set of sensors $\chi$ in $\Delta $.
Recall that edges are added when  sensors are within $r_b$ of each other,
and 2-simplices are added wherever all three edges have already been included.

Consider all possible subsets $A \subset \chi_{int}$ and
as before construct the Rips complex $\RR_A$,
the largest subcomplex of $\RR$ which does not utilize the nodes in $A$.
The collection of  these Rips complexes forms a poset under inclusion,
where $A \subset B$ gives the reverse inclusion $\RR_B \subset \RR_A$.
Arrange all of these Rips complexes into a Hasse diagram,
as shown in Figure \ref{fig: Inclusion Diagram} for the example in
Figure \ref{fig: rips complex example}.
Here we place $\RR_A$ in the row indexed by the number of elements in $A$,
and we have shaded all the complexes $\RR_A$ which fail.
A line is drawn between $\RR_A$ and $\RR_B$ if $A$ is obtained from $B$
by removing a vertex.

\begin{figure}[h]
\begin{center}
\includegraphics[scale=.9]{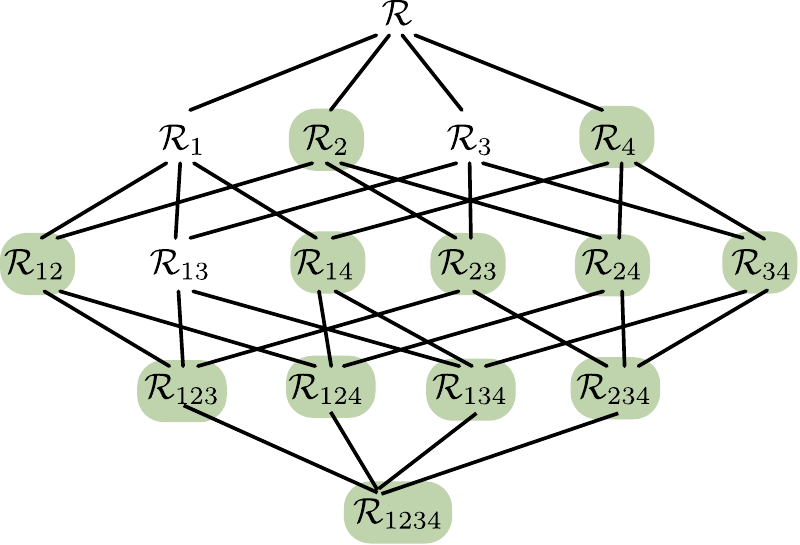}
\end{center}
\caption{Hasse Diagram.  A Rips complex $\RR_A$ is placed on a row according to the size of $A$ and lines are drawn to show inclusion between complex is neighboring rows.}\label{fig: Inclusion Diagram}
\end{figure}

In Figure \ref{fig: Inclusion Diagram} notice that, if $\RR_B$ fails the dS-G criterion, the Rips
complexes for all its supersets of $B$ do as well, so all its successors are also shaded.
This means that when searching for failures we do not have to check every possible subset.
Using breadth first search from $\RR$, we only need check complexes where all the predecessors are cake sets.
From this search pattern, if it is necessary to check the set in the first place
then it is not just a death set but a minimal death set.
This means that there is no post processing needed to determine the list of minimal death sets.

Given this setup, we now consider the probability of failure of the dS-G criterion.


\subsection{Probability of Failure} \label{S: Probability of Failure}

Let $X_i$ be a random variable which gives the time of death of node $v_i$.
In many cases, $X_i$ will be an exponential random variable,
but this has no effect on our result so we make no such assumption.
We do however assume that $X_i$ and $X_j$ are independent for $i \neq j$.

Let $S_A$ be the random variable which gives the first time at which all nodes in the set have
failed,
clearly $S_A=\max\{X_1, X_2, \cdots, X_\sizeset\}$.
Because the failures of the nodes are independent events, we have
\begin{align*}
\P(S_A\leq t) &= \P(\max\{X_1,X_2,\cdots, X_\sizeset\}\leq t)\\
	&= \P(X_1\leq t)\P(X_2\leq t)\cdots \P(X_\sizeset \leq t).
\end{align*}
Next, let $\DD = \{A_1, \cdots, A_\death \}$ be the collection of all death sets, not necessarily minimal.
Let $C$ be the random variable which gives the time of failure of the dS-G criterion for the system.
The value of $C$ gives the first time that all of the nodes in one of the $A_i$ have failed,
i.e. $C=\min\{S_{A_1},\cdots, S_{A_\death}\}$.
Hence the probability that the system has failed the dS-G criterion by time $t$ is given by
\begin{equation*}
	\P(C\leq t)=\P(\min\{S_{A_1},\cdots,S_{A_\death}\}\leq t).
\end{equation*}
Unfortunately, these events are not independent since many of the death sets have non-trivial intersections.
On the other hand, we can organize $\DD$ using the concept of death chains.

A \textbf{death chain} is a sequence
$A_1 \subset A_2 \subset \cdots \subset A_q$ of death sets $A_i \in \DD$.
A \textbf{maximal death chain} is a death chain which cannot be increased in length, either by inserting any
intermediate set between elements of the chain, or adding any sets to either end of the chain.
Note that such a chain starts with a minimal death set,
ends with $A_\chiint = \chi_{int}$, and the number of elements increases by exactly one going from
$A_i$ to $A_{i+1}$.
If we read off the Hasse diagram as in Figure \ref{fig: Inclusion Diagram},
we see that a maximal chain is a path which goes from a minimal death set to the bottom of the diagram.

We will call a maximal death chain with $A$ as the minimal element an \textbf{$A$-chain}.
This allows consideration of all possible maximal chains from the complex $\RR$,
grouped by first element: in our case $A_1$-chains through $A_\mindeath$-chains.
Using this organization of the death sets gives the following theorem:

\begin{thm}
Let $\{A_1,\cdots,A_\mindeath\}$ be the set of minimal death sets for the Rips complex $(\RF)$.
Then the probability that the complex has failed by time $t$ is equal to

	\begin{equation*}
		\P(\textrm{\emph{Failure by time }}t)= \P\left(\min_i\{S_{A_i} \} \leq t\right).
	\end{equation*}
	
\end{thm}

\begin{proof}
If $\DD$ is the set of all death sets, some set $B$ in $\DD$ must have failed in order to cause failure of the dS-G criterion.
Every set in $\DD$ contains one of the minimal death sets $A_i$,
so every set in $\DD$ is in at least one $A_i$-chain.

It is obvious that given any $A_i$-chain $\Ch$, where the time at which any of the death sets in the chain have failed is given by $Y$, satisfies
\begin{equation*}
 \P\left( Y \leq t \right) = \P(S_A\leq t).
\end{equation*}
From this we can also conclude that any set of $A$-chains $\{\Ch^1,\cdots,\Ch^k\}$, where each respective time of failure is given by $Y^1,\cdots,Y^k$, satisfies
\begin{equation*}
 \P\left( \bigcup_j \{Y^j \leq t \}\right) = \P(S_A\leq t).
\end{equation*}

Therefore, instead of asking for failure of the dS-G criterion, we can then ask for the time when at least one of the
sets in at least one of the chains has failed.
Hence
	\begin{align*}
		\P(\textrm{Failure by time }t)
			&=\P\left(\bigcup_i\{Y_{A_i}\leq t\}\right)\\
			&=\P\left(\bigcup_i\{S_{A_i}\leq t\}\right)\\
			&=\P\left(\min_i\{S_{A_i}\}\leq t\right).
	\end{align*}

\end{proof}

This means that the probability of failure of the system can be computed given the minimal death sets.
Notice that there is still work to be done since the minimal death sets may have intersections.
However, since we assume that we have a small number of sensors that are well distributed
and are not extremely dense in the domain,
the number of intersections will be small and therefore this later computation is feasible.
Thus, we seek an algorithm to compute the minimal death sets
although from our knowledge that the problem is \#P-complete,
we expect that this algorithm will be exponential in the worst case.


\subsection{Death Sets Algorithm}\label{S: Death Sets Algorithm}

When constructing an algorithm to determine the set of minimal death sets, the search space is the Hasse diagram described in Section \ref{S: Hasse}.
Since this has size $2^{|\chi_{int}|}$, we expect this is the source of our complexity issues.

Search the Hasse diagram using breadth first search.
This will exploit the property that if $B$ is a death set and $B \subset A$ then $A$ is also a death set.
Hence, to find the minimal death sets, we must only check complexes where every predecessor is still a cake set.
If we are forced to check all of the nodes of the Hasse diagram, then we will need to check $2^{|\chi_{int}|}$ complexes.
However, with our assumption that we do not have a very dense set of sensors, it should take removal of a small set of sensors in order to break the dS-G criterion.
This means the size of minimal death sets will be relatively small, and thus they will be close to the top of the Hasse diagram.
More importantly, it means there will be relatively few of them so out output size will not be too large.

Given this method to work through the sets, we need an efficient way
to check the dS-G criterion.
Consider the subcomplex $\RR_A$, thought of as the point in a filtration of
$\RR$ where all simplices have been added except those which have vertices in $A$.
Order the simplices so that those in the fence come first in the ordering.
Initial intuition says that in order to talk about failure of the dS-G criterion when the nodes in $A$ are removed,
we should filter $\RR$ so that all nodes, edges,
and triangles which have any vertex in $A$ are last in the ordering.
We could then construct the boundary matrix for this ordering,
cut off the final columns corresponding to simplices which would be gone if the vertices in $A$ failed,
and reduce the resulting matrix in order to read off the homology of $(\RR_A,\FF)$.

However, this turns out to be much more work than is needed.
If we add all degree $1$ and degree 0 simplices from the start,
even if they have interior nodes which we assume
to have failed,  the failure of the dS-G criterion is not affected.
The following expresses this and is elementary to prove.
\begin{lemma}
Let $Y$ be a 1-dimensional simplicial complex,
$X$ a 2-dimensional simplicial complex whose vertex set may intersect nontrivially with $Y$, and $\FF\subset X$, then
$H_2(X,\FF) \iso H_2(X \cup Y,\FF)$.
	\label{L: Extra Frame Stuff}
\end{lemma}

This implies that the time in the filtration when we add the 1-simplices is irrelevant to the homology group we are interested in,
namely $H_2(\RF)$.
Thus, we can order our filtration so that all the 2-simplices are at the end
and consider the failure of a node as the failure only of the 2-simplices which contain it as a face.

Given this filtration, we reduce the matrix $D$ via the persistence algorithm (see, e.g., \cite{Edelsbrunner2010}),
and consider the rightmost columns, which correspond to 2-simplices.
If the row for a simplex has no lowest one in a row below those corresponding
to the fence simplices, the addition of that 2-simplex creates a new class in $H_2(\RR,\FF)$.
If this column is not entirely 0 and has a lowest 1 in a row corresponding to a fence simplex,
that class has a boundary which is nonzero in $H_1(\FF)$.
Thus, our dS-G criterion reduces to looking for a column corresponding to a 2-simplex which has a
lowest one in a row corresponding to a fence simplex.

This shows that we can quickly determine  whether a complex satisfies the dS-G criterion  once we have
determined the correct filtration for $\RR_A$ and have reduced the matrix $D$.
We would like to not have to rewrite and re-reduce the matrix $R$ for each complex $\RR_A$ to be checked.
So, let us determine an efficient way to swap all the 2-simplices that
have a vertex in our failure set $A$ to the end.
For this, we turn to \cite{Cohen-Steiner2006}, which gives an algorithm to quickly update and
maintain the properties of $R$ and $U$, where $D=RU$, as we swap columns (Notice that $U=V\inv$ from the earlier discussion in Section \ref{S: Persistent Homology}).
This means that we do not have to  rerun the persistence algorithm each time to reduce the matrix $D$.

Let $D=RU$ be an $RU$-decomposition.
That is, $R$ is the reduced matrix, and $U$ is upper triangular.
Let $P$ be the matrix which swaps rows $i$ and $i+1$, so that $PDP$ is the boundary matrix with simplices
$\sigma_i$ and $\sigma_{i+1}$ switched.
This can be written as $PDP=(PRP)(PUP)$, so we need to determine when $PRP$ is not reduced and
$PUP$ is not upper triangular.
Notice that $PRP$ is not reduced if and only if there are columns $k$ and $l$ with
$\low_R(k)=i$, $\low_R(l)=i+1$, and $R[i,l]=1$. $PUP$ is not upper triangular if and only if $U[i,i+1]=1$.

In \cite{Cohen-Steiner2006}, cases are split into whether $\sigma_i$ and $\sigma_j$ are positive or negative.
The one case that will not occur for us is the possibility that $\sigma_i$ and $\sigma_{i+1}$ are positive,
and there are rows $k$ and $l$ with $\low_R(k)=i$ and $\low_R(l)=i+1$.
This would imply that $\sigma_l$ and $\sigma_k$ are 3-simplices whose additions kill the classes
born by the addition of $\sigma_i$ and $\sigma_{i+1}$.
As we are assuming $\RR$ is a 2-dimensional simplicial complex, this case is impossible, hence we can disregard it.
With respect to the other cases, we can either swap rows and columns in $R$ and $U$,
hence replace them with $PRP$ and $PUP$ with no issues, or we must replace them with $PRWPW$
and $WPWUP$, where $W$ is the matrix which adds column $i$ to column $i+1$ in $R$.
We replace $R$ and $U$ with $PRP$ and $PUP$ if
\begin{list}{$\diamondsuit$}{}
\item $\sigma_i$ and $\sigma_{i+1}$ are both positive simplices,
\item $\sigma_i$ and $\sigma_{i+1}$ are both negative simplices and $$U[i,i+1]=0,$$
\item $\sigma_i$ is negative, $\sigma_{i+1}$ is positive, and $$U[i,i+1]=0,$$
\item $\sigma_i$ is positive and $\sigma_{i+1}$ is negative.
\end{list}
We instead replace $R$ with $PRWPW$ and $U$ with $WPWUP$ if
\begin{list}{$\diamondsuit$}{}
\item $\sigma_i$ and $\sigma_{i+1}$ are both negative and $$U[i,i+1]=1$$ 
\item $\sigma_i$ is negative, $\sigma_{i+1}$ is positive, and $$U[i,i+1]=1.$$
\end{list}
Hence,  the columns of  2-simplices which correspond to failure of specific vertices can be quickly swapped to the end of our matrix.
Now that matrix $R$ corresponding to the filtration placing the simplices in $A$ at the end has been reduced, we can check whether our complex $\RR_A$ passes the dS-G criterion.
In the language of persistence homology, the dS-G criterion is looking for a column which represents a positive simplex, and therefore a simplex which adds a new class to  $H_2(\RR_A,\FF)$.
Additionally, the boundary of this new class is nonzero in $H_1(\FF)$.
In terms of the matrices, we need to find a column $i$ in the reduced matrix $R$ which is   has a lowest 1 corresponding to a simplex in $\FF$.
A representation of this column is in Figure \ref{F: matrixRfundamental}.
\begin{figure}
\begin{center}
\includegraphics{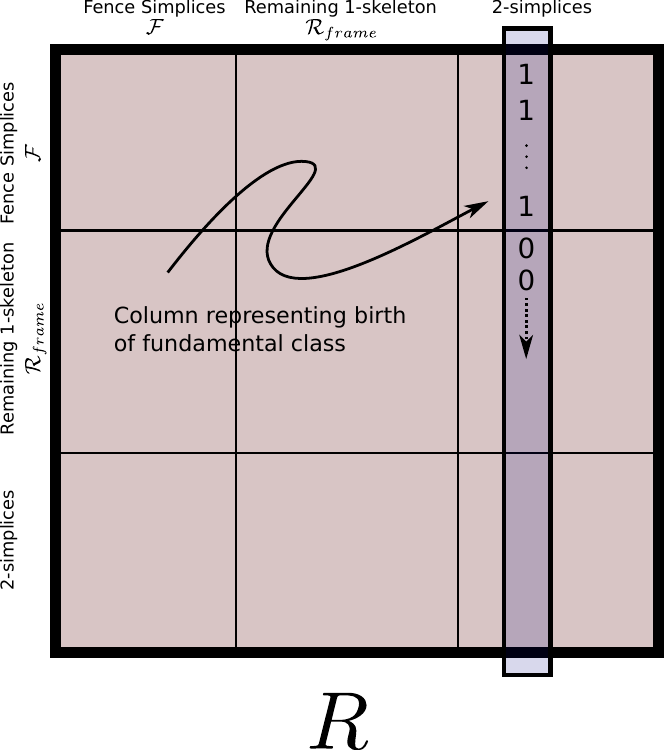}
\end{center}
\caption{The column representing the birth of a fundamental class. It corresponds to a 2-simplex and its lowest
$1$ corresponds to a simplex in $\FF$.} \label{F: matrixRfundamental}
\end{figure}
Our algorithm is as follows.

\begin{tabbing}
\\
\textbf{CakeOrDeath}($\RF$)\\
\\
Given: \= The boundary matrix $D$ with filtration \\
\> $\{$Fence, Remaining 1-simplices, 2-simplices $\}$ \\
Reduce $RU=D$\\
\textbf{for} \=$A \in$\=$ \chi_{int}$,  in the order of BFS in the Hasse diagram:\\
\>	Swap all columns corresponding to 2-simplices with\\
\>\>		a vertex in $A$ to end of matrix $D$, and maintain\\
\>\>		$R'$ and $U'$.\\
\>	\textbf{if} there is not a column as in Figure \ref{F: matrixRfundamental}:\\
\>\>		mark $A$ as `Minimal Death.'\\
\>	\textbf{else}:\\
\>\>		mark $A$ as `Cake.'\\
\>	\textbf{endif}\\
\textbf{endfor}
\end{tabbing}


\subsection{Complexity of Algorithm} \label{S: Complexity of Algorithm}
As a beginning aside, we point out why we chose breadth instead of depth first search.
BFS has the property that we will only ever check sets which are cake or minimal death.
On the other hand, DFS would require post-processing to determine which of the death sets  found were minimal death sets.
The perk of DFS, however, is that it requires less matrix swaps since multiple sets $A$ can be labeled as cake or death by reading off of one matrix.
As we wish to have less post-processing, we choose to use BFS for our algorithm.

Assume that the matrix $D$ is stored as a sparse matrix.
This is done with a linear array of lists $D[1,\cdots,m]$ where $m$ is the total number of simplices
in the 2-dimensional complex $\RR$.
Also, assume the ordering of the simplices is
$\left\{ \textrm{fence}, \textrm{vertices and remaining 1-simplices}, \textrm{2-simplices} \right\}$.
Each entry $D[i]$ in this array stores a linked list denoting the locations of the codimension-1 faces of $\sigma_i$, or equivalently, the 1s in column $i$ of the full matrix $D$.  This not only speeds up the operations required on the matrix, but reduces the storage size of $D$ to $O(m)$.

There are four major parts to the algorithm.
The first is to reduce $D=RU$.
Using the persistence algorithm, this takes time at most $O(m^3)$.
See \cite{Edelsbrunner2010} for details.

For each complex to be checked,  all necessary columns must be swapped to the right side of matrix $R$.
If there are $t$ 2-simplices, at worst there are $t^2/4$ swaps to be performed.
While each swap has at worst an $O(m)$ running time, from \cite{Cohen-Steiner2006}
there is an amortized time proportional to the number of 1s in the affected rows and columns, so this is $O(1)$.
This step therefore has an amortized cost of $O(t^2/4)$.

Next, we check for a fundamental class.
If this is done in $N$ complexes,  then $N \leq 2^{\chi_{int}-1}$.
(Recall that because we assume that the dense set of sensors is not dense, $N$ will likely be much smaller than $2^{|\chi_{int}|}$.)

If  a vector \texttt{Low} giving the location of the lowest 1 in each column is maintained throughout the process of swapping,  easily done via the cases in \cite{Cohen-Steiner2006}, it only takes $O(t)$ time to check for a column which fits our requirements.
Hence for each death set, the amortized running time is $O(t^2)$, and so the overall running time is $O(m^3+Nt^2)$.

To make this running time feasible, one needs to keep $N$ under control.
The easiest way to do this is to have a sparse set of sensors.
For example, suppose the area of the domain $\Delta$ is $R$ and we have $n$ sensors.
Each sensor covers an area of $\pi r_c^2$, so the sensors cover a total area (double counting overlap) of $n\pi r_c^2$.
This means the expected number of sensors covering any point is $x = n\pi r_c^2 / R$.
Therefore the death sets should be of size approximately $x$, so it is only necessary
to check complexes through about the $x^{\textrm{th}}$ row.
Thus $N \approx \binom{n}{1} +\cdots + \binom{n}{x}$.

In conclusion, computing the probability exactly will be easier if the set of sensors is sparse,
but what is gained in exactness of the computation is lost  in the robustness of the network.


\section{A Dynamic Algorithm for a Monitored System}\label{S: Monitored System}

Suppose we are in a situation where the deterministic algorithm is not feasible.
Computing the probability of failure exactly is an NP-hard problem,
and thus an exact computation is essentially impossible when the set of sensors is large.
Instead, assume that a central monitoring station receives information as to whether or not each sensor has failed.
In this case, a more practical question is to ask when  the system is getting close to failure and so we seek a dynamic algorithm to predict which nodes would cause failure of the criterion should they fail soon.
To do this, we will create a new criterion built from the old which will give an early warning for failure.
It essentially gives a flag on each interior vertex warning that its failure would probably cause failure of the dS-G criterion.

For technical reasons, we will assume  in this section  that the domain is convex.
A domain that is not convex can have a radius where the Rips complex has a nontrivial class in $H_1(\RR)$ even though $(\RF)$ passes the dS-G criterion.
This assumption is much stronger than is necessary since all we really need
is  that $H_1(\RR)=0$ whenever $(\RF)$ passes the dS-G criterion.

The main idea for the new criterion comes from the following theorem.
\begin{thm}\label{Thm: New Criterion}
 Assume that the pair of simplicial complexes $(\RF)$ passes the dS-G criterion, $H_2(\RR)=0$, and $w$ is a vertex in $\RR-\FF$.
Then $(\RR_w,\FF)$ passes the dS-G criterion if and only if $H_1(\Lk(w))=0$.
\end{thm}

\begin{proof}
Assume that $(\RR_w,\FF)$ passes the dS-G criterion and that the domain is convex, which
implies that $H_1(\RR_w)$ is trivial.
Mayer-Vietoris  for $\RR = \RR_w \cup \overline{\St}( w)$ gives the exact sequence
\begin{equation*}
 \xymatrix{
  H_2(\RR) \ar[r] &  H_1(\Lk (w)) \ar[r] & H_1(\RR_w) \oplus H_1(\overline{\St} (w,\RR))
}
\end{equation*}
Thus, since $H_2(\RR)$, $H_1(\RR_w)$ and $H_1(\overline{\St} (w,\RR))$ are all  trivial,  $H_1(\Lk (w)) = 0$.

Assume on the other hand that $H_1(\Lk(w))=0$.
Note that if $(\RF)$ passes the dS-G criterion, $\widetilde{H}_0(\Lk(w)) = 0$.
Using the long exact sequence of the pair $(\RR_w,\Lk(w))$, we have
\begin{equation*}
\xymatrix{
0 \ar[r] &\widetilde{H}_1(\RR_w) \ar[r] & \widetilde{H}_1(\RR_w,\Lk(w,\RR)) \ar[r] & 0
}
\end{equation*}
hence the middle map is an isomorphism.

By excision, $H_1(\RR_w,\Lk(w,\RR)) \iso H_1(\RR,\overline{\St}(w,\RR))$.
Using the knowledge that $\overline{\St}(w,\RR)$ has the homotopy type of a point for the first isomorphism and the long exact sequence of the pair $(\RR,\cdot)$ for the second, we have
\begin{equation*}
 H_1(\RR,\overline{\St}(w,\RR)) \iso H_1(\RR,\cdot) \iso H_1(\RR).
\end{equation*}
Thus, $H_1(\RR_w) \iso H_1(\RR)$.

Consider the diagram
\begin{equation*}
 \xymatrix{
H_2(\RR_w,\FF) \ar[r] \ar[d] & H_1(\FF) \ar[r] \ar[d]^\iso & H_1(\RR_w) \ar[d]^\iso\\
H_2(\RR,\FF) \ar[r] & H_1(\FF) \ar[r] & H_1(\RR)
}
\end{equation*}
Since $(\RF)$ passes the dS-G criterion, there is a fundamental class $\alpha$ in $H_2(\RF)$.
By definition, it maps to $\partial \alpha$ which is nonzero in $H_1(\FF)$, and since the two rows are exact, $\partial \alpha$ maps to 0 in $H_1(\RR)$.
The last two vertical maps are isomorphisms, so there must be a nonzero $\beta$ in $H_2(\RR_w,\FF)$ which maps to $\partial \alpha$ under the top left horizontal map, and therefore $(\RR_w,\FF)$ passes the dS-G criterion.

\end{proof}

Notice that the assumption that $H_2(\RR)=0$ is only used  for one direction of the theorem: if $(\RR_w,\FF)$ passes, then we have a link with trivial first homology.
Despite being counterintuitive, it is possible for a set of points in the plane to have a non-trivial second homology group \cite{Chambers2009}.
The expectation is that this event will not be frequent, but it must be kept in mind as we create a new criterion in this monitored set up.

\subsection{The new criterion and complexity}

Given Theorem \ref{Thm: New Criterion}, we propose a new criterion to complement the de Silva - Ghrist criterion:
\begin{defn}[Link Condition]
 If an interior vertex $w$ has $H_1(\Lk(w))\neq 0$, we say it is flagged.  Otherwise, we say it is not flagged.
\end{defn}
The idea is that if a vertex is flagged, there is a chance its removal will cause failure of the dS-G criterion.  If it is not flagged, then its removal can do no harm.

With this definition in mind, we give a dynamic algorithm to follow as nodes fail:

\begin{tabbing}
\\
\textbf{MonitoredSystemFailure}($\RF$)\\
\\
Given: Simplicial complex pair $(\RF)$\\
\\
Check \= dS-G criterion 
 (We assume that this initial \\ \>check will always pass)\\
Compute link of each vertex $w$ and $H_1(\Lk(w))$\\
\textbf{if} \={$H_1(\Lk(w)) = 0$}:\\
\>	Mark $w$ as flagged.\\
\textbf{endif}\\
\\
\textbf{if} \=ver\=tex \=$v$ fails:\\
\>		\textbf{if} {$v$ is flagged}:\\
\>\>			Update matrix $R$ to remove dead simplices\\
\>\>			Check dS-G criterion\\
\>\>			 \textbf{if} $(\RR_w,\FF)$ fails the dS-G criterion:\\
\>\>\>				 Break\\
\>\>			\textbf{endif} \\
\>		\textbf{endif} \\
\>		Update links of vertices\\
\>		Compute $H_1(\Lk(w))$ for $w$ whose link has \\ \>\>changed\\
\>		Mark or unmark $w$ as flagged according to $H_1$ \\ \>\>computation.\\
\textbf{endif} \\		
\\
\end{tabbing}

This algorithm turns out to be polynomial in the number of simplices in the worst case.
We will split the complexity computation into two parts: the initialization step, done before any vertex has failed, and the time taken for the algorithm for each failed vertex.

\subsubsection*{Initialization}

Let $m$ be the number of two simplices of dimension $\leq 2$.  In section \ref{S: 2-skeleton}, we showed that this is the highest dimension simplex needed for the dS-G criterion, and since the link condition only looks at the first dimension of the complex, nothing above the second dimension is required.

The complexity of computing the link of $w$ is directly related to the number of simplices containing $w$ as a vertex.
In fact, given a list of all simplices in $\RR$ which include vertex $w$, print the simplex obtained by removing vertex $w$ from the simplex.
This is the link, so given a listing of the simplices with vertex $w$, the link of $w$ can be computed in time $O(k)$ where $k$ is the number of adjacent simplices.
Since $k$ is obviously less than $m$, the time to initially compute all the links is $O(mn)$.

The time to compute $H_1(\Lk(w))$ is $O(k^3)$, again with $k$ equal to the number of simplices in the link of $w$.  Since, $k<m$, the time taken for this initial step is $O(nm^3)$.
Thus, the entirety of the initialization step takes time $O(m^3+mn+nm^3) = O(nm^3)$.

\subsubsection*{Failed Vertex}

In the worst case, every failed vertex is flagged and so the dS-G criterion must be recomputed each time.
As seen in section \ref{S: Complexity of Algorithm}, updating the matrix $R$ and checking the dS-G criterion takes time $O(t^2)$ where $t$ is the number of two simplices.

What is interesting about the link condition is that it is easy to maintain the links of all the interior vertices.
  Let $\Lk(\sigma,X)$ be the link of $\sigma$ in the simplicial complex $X$ and let $X_w$ be the largest subcomplex of $X$ without the vertex $w$.

\begin{lemma}
For any vertices $v,w$ in a simplicial complex $X$,
\begin{equation}\label{E: Link Condition}
 \Lk(v,X_w) = \Lk(v,X) \cap X_w.
\end{equation}
\end{lemma}
\begin{proof}
 If $\sigma \in \Lk(v,X_w)$ then obviously $\sigma \in X_w$.
 Also, this implies that the simplex $\tau = <\sigma,v> \in X_w$.
 Since $\sigma< \tau$ and $v \not \in \sigma$, we must have that $\sigma \in \Lk(v,X)$, so $\Lk(v,X_w) \subset \Lk(v,X) \cap X_w$.

 Let $\sigma \in \Lk(v,X) \cap X_w$. Since it is in $\Lk(v,X)$, the simplex, $<\sigma,v> \in X$.
 As $\sigma$ is also in $X_w$, $w \not \in \sigma$, so $<\sigma,v> \in X_w$.  Therefore, $\sigma \in \Lk(v,X_w)$, and equation \ref{E: Link Condition} follows.
\end{proof}

This lemma implies that the only  update needed after the failure of a vertex is to delete any simplices in the link which were also deleted in the simplicial complex.
In the worst case, the link of every vertex must be updated, the first homology recomputed, and the flag remarked as needed.
Since the size of each link is at most $m$, this step takes $O(n(m+m^3)) = O(nm^3)$.

If this sequence of events happens for every $n$, this second part of the algorithm takes time $$O\bigg(n(t^2+t+mn+m^3n) \bigg) = O(m^3n^2).$$
Combining this with the initializing step, the whole algorithm takes at worst time $O(m^3n^2)$, so is polynomial in the number of simplices.


\section{Conclusions and Possible Extensions}\label{S: Conclusion}

In this paper, we have extended the problem posed by de Silva and Ghrist in \cite{DeSilva2006}
by assuming that  sensors have a probability of failure, and asking for the probability of failure of the dS-G criterion for
coverage.
We determined that the generalized version of the problem is \#P-complete, and thus it is unlikely that there is an
 algorithm to answer this question in general which runs in a reasonable amount of time.
Finally, we provided a deterministic algorithm which does work in the case of a small set of sensors,
and a method to predict failure when the system is larger but is being monitored.

The obvious immediate extension of our work is to determine whether the the version of the problem
posed by de Silva and Ghrist in \cite{DeSilva2007}, which allows for higher dimensions and looser
boundary conditions, is also amenable to an application of probability of failure.
We conjecture that this extended problem will also be NP-hard.

In the long term, we would like to see more applications of computational topology to the
design and analysis of sensor networks.
Since we can make such strong conclusions with such weak assumptions on the capabilities of the sensors,
we expect that such applications are abundant.


\bibliographystyle{authordate1}	
\bibliography{FailureFiltrationLibrary}

\begin{thebibliography}{10}

\bibitem{Agarwal2006}
PK~Agarwal, H~Edelsbrunner, J~Harer, and Y~Wang.
\newblock Extreme elevation on a 2-manifold.
\newblock {\em Discrete \& Computational Geometry}, 36:553--572, 2006.

\bibitem{Ban2004}
YA~Ban, H~Edelsbrunner, and J~Rudolph.
\newblock Interface surfaces for protein-protein complexes.
\newblock In {\em Proceedings of the eighth annual international conference on
  Resaerch in computational molecular biology}, RECOMB '04, pages 205--212, New
  York, NY, USA, 2004. ACM.

\bibitem{Carlsson2009}
G~Carlsson.
\newblock {Topology and data}.
\newblock {\em Bulletin of the American Mathematical Society}, 46(2):255--308,
  January 2009.

\bibitem{Carlsson2008}
G~Carlsson, T~Ishkhanov, V~de~Silva, and A~Zomorodian.
\newblock On the local behavior of spaces of natural images.
\newblock {\em International Journal of Computer Vision}, 76:1--12, 2008.
\newblock 10.1007/s11263-007-0056-x.

\bibitem{Chambers2009}
EW~Chambers, V~de~Silva, J~Erickson, and R~Ghrist.
\newblock {Vietoris-Rips Complexes of Planar Point Sets}.
\newblock {\em Discrete \& Computational Geometry}, 44(1):75--90, July 2009.

\bibitem{Chazal2009}
F~Chazal, LJ~Guibas, SY~Oudot, and P~Skraba.
\newblock Analysis of scalar fields over point cloud data.
\newblock In {\em Proceedings of the twentieth Annual ACM-SIAM Symposium on
  Discrete Algorithms}, SODA '09, pages 1021--1030, Philadelphia, PA, USA,
  2009. Society for Industrial and Applied Mathematics.

\bibitem{Cohen-Steiner2006}
D~Cohen-Steiner, H~Edelsbrunner, and D~Morozov.
\newblock {Vines and vineyards by updating persistence in linear time}.
\newblock {\em Proceedings of the twenty-second annual symposium on
  Computational geometry - SCG '06}, page 119, 2006.

\bibitem{Colbourn1987}
CJ~Colbourn.
\newblock {\em The Combinatorics of Network Reliability}.
\newblock Oxford University Press, 1987.

\bibitem{DeSilva2006}
V~de~Silva and R~Ghrist.
\newblock {Coordinate-free Coverage in Sensor Networks with Controlled
  Boundaries via Homology}.
\newblock {\em The International Journal of Robotics Research},
  25(12):1205--1222, December 2006.

\bibitem{DeSilva2007}
V~de~Silva and R~Ghrist.
\newblock {Coverage in sensor networks via persistent homology}.
\newblock {\em Algebraic \& Geometric Topology}, 7:339--358, April 2007.

\bibitem{Edelsbrunner2009}
H~Edelsbrunner and J~Harer.
\newblock The persistent morse complex segmentation of a 3-manifold.
\newblock In Nadia Magnenat-Thalmann, editor, {\em Modelling the Physiological
  Human}, volume 5903 of {\em Lecture Notes in Computer Science}, pages 36--50.
  Springer Berlin / Heidelberg, 2009.

\bibitem{Edelsbrunner2010}
H~Edelsbrunner and J~Harer.
\newblock {\em Computational Topology: An Introduction}.
\newblock American Mathematical Society, 2010.

\bibitem{Garey1979}
MR~Garey and DS~Johnson.
\newblock {\em Computers and Intractibility: A Guide to the Theory of
  NP-Completeness}.
\newblock W.H. Freeman and Co., 1979.

\bibitem{Ghrist2005}
R~Ghrist and A~Muhammad.
\newblock {Coverage and hole-detection in sensor networks via homology}.
\newblock {\em IPSN 2005. Fourth International Symposium on Information
  Processing in Sensor Networks, 2005.}, (1):254--260, 2005.

\bibitem{Headd2007}
JJ~Headd, YEA Ban, P~Brown, H~Edelsbrunner, M~Vaidya, and J~Rudolph.
\newblock Protein-protein interfaces:  properties, preferences, and
  projections.
\newblock {\em Journal of Proteome Research}, 6(7):2576--2586, 2007.
\newblock PMID: 17542628.

\bibitem{Tahbaz-Salehi2010}
A~Tahbaz-Salehi and A~Jadbabaie.
\newblock Distributed coverage verification in sensor networks without location
  information.
\newblock {\em Automatic Control, IEEE Transactions on}, 55(8):1837 --1849,
  aug. 2010.

\bibitem{Valiant}
LG~Valiant.
\newblock The complexity of computing the permanent.
\newblock {\em Theoretical Computer Science}, 8(2):189 -- 201, 1979.

\bibitem{Yick2008}
J~Yick, B~Mukherjee, and D~Ghosal.
\newblock Wireless sensor network survey.
\newblock {\em Computer Networks}, 52(12):2292 -- 2330, 2008.

\end{thebibliography}

\end{document}